\documentclass{llncs}
\usepackage{graphics,graphicx,latexsym}
\usepackage{amssymb}
\usepackage{enumerate}

  \newcommand{\myskip}{\vskip 0.5cm \noindent}

 \newif\iffinal
 \finaltrue
  \def\timestamp{%
      \the\year--%
      \ifnum\month<10 0\fi\the\month--%
      \ifnum\day<10 0\fi\the\day\ {%
        \count9=\time \divide\count9 by 60 %
        \ifnum\count9<10 0\fi\the\count9%
        \count7=\time \multiply\count9 by60 \advance\count7 by-\count9%
        :%
        \ifnum\count7<10 0\fi\the\count7%
      }
  }

  \iffinal
  \else
  \usepackage{fancyheadings}
  \fi

  \newcommand{\pr}{\textsf{root}}
  \newcommand{\nth}{\mbox{${}^{\textsl{\scriptsize th}}$}}
  \newcommand{\nd}{\mbox{${}^{\textsl{\scriptsize nd}}$}}
  \newcommand{\rd}{\mbox{${}^{\textsl{\scriptsize rd}}$}}
  \newcommand{\dotdot}{\ldotp\ldotp}

  \newcommand{\squares}{\textsf{squares}}
  \newcommand{\cubes}{\textsf{cubes}}

\begin{document}
  \date{}
  \author{
    Marcin Kubica\inst{1}
    \and
    Jakub Radoszewski\inst{1}
    \and
    Wojciech Rytter\inst{1,2}
    \and \vskip 0.2cm
    Tomasz Wale\'n\inst{1}
  }

  \institute{
    Department of Mathematics, Computer Science and Mechanics, \\
    University of Warsaw, Warsaw, Poland\\
    \email{\{kubica,jrad,rytter,walen\}@mimuw.edu.pl}
    \and
    Faculty of Mathematics and Informatics,\\
    Copernicus University, Toru\'n, Poland
  }

  \title{On\ the\ maximal\ number\ of\ cubic \\ subwords\   in\ a\ string\thanks{
      Supported by grant N206 004 32/0806
      of the Polish Ministry of Science and Higher
      Education.
    }
}
  \maketitle

  \begin{abstract}
    We investigate the problem of the maximum number of cubic subwords
    (of the form $www$) in a given word.
    We also consider square subwords (of the form $ww$).
    The problem of the maximum number of squares in a word is not well understood.
    Several new results related to this problem are produced in the paper.
    We consider two simple problems related to the maximum number of subwords which
    are squares or which are highly repetitive;
    then we provide a nontrivial estimation for the number of cubes.
    We show  that the maximum number of   squares $xx$ such that $x$ is not a primitive word
    (nonprimitive squares) in a word of length $n$ is exactly
    $\left\lfloor \frac{n}{2}\right\rfloor - 1$, and the maximum number of subwords
    of the form $x^k$, for $k\ge 3$, is exactly $n-2$.
    In particular, the maximum number of cubes in a word is not greater than $n-2$ either.
    Using very technical properties of occurrences of cubes, we  improve this bound significantly.
    We show that the maximum number of cubes in a word of length $n$ is between
    $\frac{1}{2}n$ and $\frac{4}{5}n$ \footnote{
    In particular, we improve the lower bound from the conference version of the
    paper \cite{Iwoca}.}.
  \end{abstract}

  \section{Introduction}
    A repetition is a word composed (as a concatenation) of several copies of another word.
    The exponent is the number of copies.
    We are interested in natural exponents higher than 2.
    In \cite{Szilard} the authors considered also exponents which are not integer.

    In this paper we investigate the bounds for the maximum number of highly
    repetitive subwords in a word of length $n$.
    A word is highly repetitive iff it is of the form $x^k$ for some integer $k$ greater than 2.
    In particular, cubes $w^3$ and squares $x^2$ with nonprimitive $x$ are highly repetitive.

    The subject of computing maximum number of squares and repetitions in words
    is one of the fundamental topics in combinatorics on words
    \cite{Karhumaki,Loth05} initiated by A.~Thue \cite{Thue},
    as well as it is important in other areas:
    lossless compression, word representation, computational biology etc.

    The behaviour of the function $\squares(n)$ of maximum number of squares in a word of length $n$
    is not well understood, though the subject of squares was studied by many authors, see
    \cite{DBLP:journals/algorithmica/CrochemoreR95,Jewels,DBLP:journals/tcs/IliopoulosMS97,marcinpiatk}.
    The best known results related to the value of $\squares(n)$ are, see
    \cite{fraenkel-simpson,ilie:2n,ilie:2n-logn}:
    $$n-o(n) \le \squares(n) \le 2n-O(\log n)\ .$$
    In this paper we concentrate on larger powers of words and show that in this
    case we can have much better estimations.
    Let $\cubes(n)$ denote the maximum number of cubes in a word of length $n$.
    We show that:
    $$\frac{1}{2}n  \le \cubes(n)\ \le \frac{4}{5}n\ .$$

    There are known efficient algorithms for the computation of integer powers in words,
    see~\cite{DBLP:journals/tcs/ApostolicoP83,DBLP:journals/ipl/Crochemore81,%
      damanik,DBLP:journals/dam/Main89,DBLP:journals/jal/MainL84}.

    The powers in words are related to maximal repetitions, also called {\em runs}.
    It is surprising that the bounds for the number of runs are much tighter than for squares,
    this is due to the work of many people
    \cite{baturo,DBLP:journals/jcss/CrochemoreI08,DBLP:conf/cpm/CrochemoreIT08,%
      DBLP:conf/lata/Giraud08,DBLP:conf/focs/KolpakovK99,DBLP:conf/fct/KolpakovK99,%
      DBLP:journals/tcs/PuglisiSS08,DBLP:conf/stacs/Rytter06,DBLP:journals/iandc/Rytter07}.

    Our main result is a new estimation of the number of cubic subwords.
    We use a new interesting technique in the analysis: the proof of the upper bound is
    reduced to the proof of an invariant of some abstract algorithm (in our invariant lemma).
    There is still some gap between upper and lower bound but it is much smaller than
    the corresponding gap for the number of squares.

    \begin{figure}
    \centering
    \includegraphics[width=.8\textwidth]{cubes30.1}
    \caption{Example of a word with 11 distinct cubes. 
        This is a word of length 30 with maximal number of cubes among binary
        words of the same length.
    }
    \end{figure}

  \section{Periodicities in strings}
    We consider \emph{words} over a finite alphabet $A$, $u\in A^*$;
    by $\varepsilon$ we denote an empty word.
    The positions in a word $u$ are numbered from $1$ to $|u|$.
    For $u=u_1\ldots u_k$, by $u[i\dotdot j]$ we denote a \textit{subword}
    of $u$ equal to $u_i\ldots u_j$; in particular, $u[i]=u[i\dotdot i]$.

    We say that a positive integer $p$ is a \emph{period} of a word
    $u=u_1\ldots u_k$ if $u_i=u_{i+p}$ holds for $1\le i\le k-p$.
    If $w^k=u$ ($k$ is a nonnegative integer) then we say that $u$
    is the $k\nth$ power of the word $w$.

    The \emph{primitive root} of a word $u$, denoted $\pr(u)$, is
    the shortest word $w$, such that $w^k=u$ for some positive $k$.
    We call a word $u$ \emph{primitive} if $\pr(u)=u$, otherwise
    it is called \emph{nonprimitive}.
    It can be proved that the primitive root of a word $u$ is the only
    primitive word $w$, such that $w^k=u$ for some positive $k$.

    A \emph{square} is the $2\nd$ power of some word, and
    an \emph{np-square} (a nonprimitive square) is a square of a word
    that is \textbf{not} primitive.
    A \emph{cube} is a $3\rd$ power of some word.

    In this paper we focus on the last occurrences of subwords.
    Hence, whenever we say that word $u$ \emph{occurs at position $i$} of the word $v$
    we mean its \textbf{last} occurrence, that is
    $v[i\dotdot i+|u|-1]=u$ and $v[j\dotdot j+|u|-1]\ne u$ for $j>i$.
    The following lemma is used extensively throughout the article.

    \begin{lemma}[Periodicity lemma \cite{FW:peridicity-lemma,Loth05}]
      \label{l:periodicity}
      If a word of length $n$ has two periods $p$ and $q$,
      such that $p+q \le n+\gcd(p,q)$, then $\gcd(p,q)$ is also a period of the word.
    \end{lemma}
    In this paper we often use, so called, weak version of this lemma,
    where we only assume that $p+q \le n$.

\section{Basic properties of highly repetitive subwords}

    A word is said to be \emph{highly repetitive} (hr-word) if it is a $k\nth$ power
    of a nonempty word, for $k\ge 3$.

    \begin{figure}[th]
      \label{pref-struct2}
      \begin{center}
        \includegraphics[width=10cm]{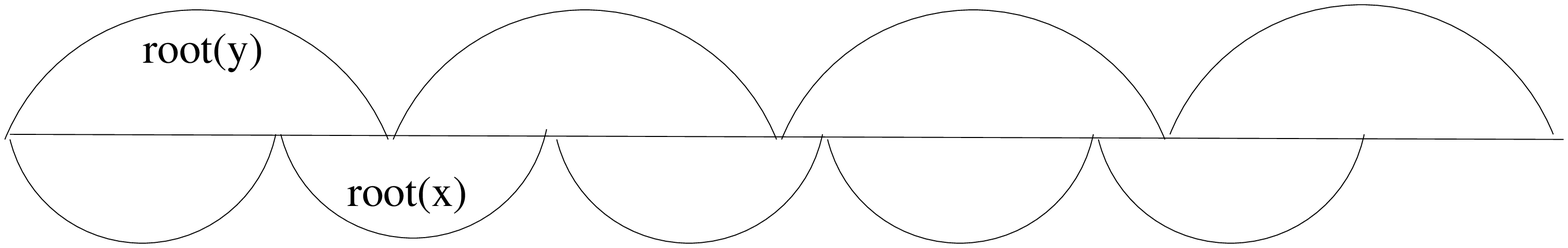}
        \caption{\label{fig1}
          The situation when one hr-word is a (long) prefix of another hr-word
          implies that $\pr(x)=\pr(y)$, consequently $x$ is a suffix of $y$.}
      \end{center}
    \end{figure}

    \begin{lemma}
      \label{l:pref-suf}
      If a hr-word $x$ is a prefix of a hr-word $y$ and $|x| \ge |y|-|\pr(y)|$,
      then $x$ is also a suffix of $y$.
    \end{lemma}

    \begin{proof}
      Due to the periodicity lemma, both words have the same smallest period
      and it is a common divisor of the lengths of their primitive roots,
      see Figure \ref{fig1}.
      Consequently, we have $\pr(x)=\pr(y)$ and
      $x$ is a suffix of $y$.
      \qed
    \end{proof}

    \begin{figure}[th]
      \label{pref-struct}
      \begin{center}
        \includegraphics[width=8.5cm]{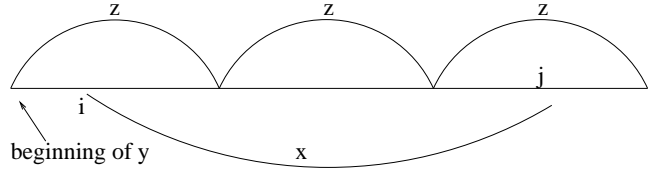}
        \caption{\label{fig2}
          The situation from Lemma~\ref{lemma3}.}
      \end{center}
    \end{figure}

    \begin{lemma}\label{lemma3}
      Assume that $x$ and $y$ are two hr-words, where $y=z^3$
      and $x$ is a subword of $y$ starting at position $i$ and ending at position $j$ such that
      $$i \ \le\  \left\lceil\frac{|\pr(z)|}2\right\rceil+1\quad \textrm{and} \quad j > |z^2|\ .$$
      Then, $|\pr(x)| = |\pr(y)|$.
    \end{lemma}
    \begin{proof}
      Let $x=w^k$, for some $k \ge 3$.
      Using the inequalities on $i$ and $j$ from the lemma, we obtain:
      $$|x| \ =\ j-i+1 \ \ge\ |z^2|+1 - \left\lceil\frac{|\pr(z)|}2\right\rceil - 1 + 1 \ \ge$$
      $$\ge\ 2 \cdot |z| - \left\lceil\frac{|z|}2\right\rceil + 1 \ \ge\ 2 \cdot |z| - \frac{|z|}2 \ =\ \frac{3}{2}\cdot|z|\ .$$
      Let us also observe that $|\pr(x)|$ and $|\pr(y)|$ are both periods of $x$.
      Moreover:
      $$|x| \ =\  |w^k|\  = \ |w|+\frac{k-1}{k}\cdot|x| \ \ge \ |w|+\frac{2}{3}\cdot|x| \ \ge$$
      $$ \ge\ |w|+|z| \ \ge \ |\pr(x)|+|\pr(y)|\ .$$
      From this, by the periodicity lemma, we obtain that $\gcd(|\pr(x)|, |\pr(y)|)$
      is also a period of $x$.
      However, $\pr(x)$ and $\pr(y)$ are subwords of $x$, so $|\pr(x)| = |\pr(y)|$,
      since in the opposite case one of the words $\pr(x),\pr(y)$ would not be primitive.
      \qed
    \end{proof}

  \section{Simple bounds for highly repetitive subwords}
    In this section we give some simple estimations of the number of square subwords
    with nonprimitive roots and cubic subwords.

    \begin{lemma}
      \label{l:bigPowers}
      Let $u$ be a word.
      Let us consider highly repetitive subwords of $u$ of the form $v^k$,
      for $k \ge 3$ and $v$ primitive.
      For each such subword we consider its (last) occurrence in $u$.
      For each position $i$ in $u$, at most one such subword can
      have its (last) occurrence at position $i$.
    \end{lemma}
    \begin{proof}
      Let us assume that we have two different hr-words $x$ and $y$ with their last occurrences
      starting at position $i$, and let us assume that $x$ is shorter.
      Then, we have $ |x|\ge |y|-|\pr(y)|$, otherwise the considered occurrence of
      $x$ would not be the last one.

      Now we can apply Lemma~\ref{l:pref-suf} --- $x$ is not only a prefix of $y$,
      but also its suffix.
      Hence, $x$ appears later in the text and the last occurrence of $x$ in $u$ does not
      start at position $i$.
      This contradiction proves that the assumption that the last occurrences of
      $x$ and $y$ start at position $i$ is false.
      \qed
    \end{proof}
    The following fact is a   consequence of Lemma~\ref{l:bigPowers}.

    \begin{theorem}
      \label{c:cubesn}
      The maximum number of highly repetitive subwords of a word of length
      $n\ge 2$ is exactly $n-2$.
    \end{theorem}
    \begin{proof}
      From Lemma~\ref{l:bigPowers} we know that at each position there can be
      at most one last occurrence of a nonempty hr-word.
      Moreover, the minimum possible length of such a word is 3.
      Therefore, there can be no such occurrences at positions $n$ and $n-1$.
      On the other hand, this upper bound is reached by the word $a^n$.
      \qed
    \end{proof}

    As a corollary, we obtain a simple upper bound for the number of cubes,
    since cubes are hr-words.
    \begin{corollary}
      Let us consider a word $u$ of length $n$.
      The number of nonempty cubes appearing in $u$ is not greater than $n-2$.
    \end{corollary}
    We improve this upper bound substantially in the next sections.
    However, it requires a lot of technicalities.
    Another implication of Theorem~\ref{c:cubesn} is a tight bound for the number of np-squares.

    \begin{theorem}
      \label{t:np-squaresn}
      Let $u$ be a word of length $n$.
      The maximum number of nonempty np-squares appearing in $u$ is exactly
      $\left\lfloor \frac{n}{2}\right\rfloor -1$.
    \end{theorem}

    \begin{proof}
      Each nonempty np-square can be viewed as $v^{2i}$ for some nonempty primitive $v$ and $i \ge 2$.
      However, each such np-square contains a subword $v^{2i-1}$, which is not an np-square
      (due to the periodicity lemma), but still a hr-word.
      Hence, the number of nonempty subwords of the form $v^{2i-1}$ (for primitive $v$ and $i \ge 2$),
      appearing in the given word, is not smaller than the number of nonempty np-squares.

      Observe that Theorem~\ref{c:cubesn} limits the total number of both
      subwords of the form $v^{2i}$ and $v^{2i-1}$ by $n-2$.

      Hence, the total number of nonempty np-squares appearing in the given word is not greater
      than $\frac{n}{2}-1$, and since it is integer, it is not greater than
      $\left\lfloor \frac{n}{2}\right\rfloor -1$.
      On the other hand, this upper bound is reached by the word $a^n$.
      \qed
    \end{proof}

\section{The structure of occurrences of cubic subwords}
    In this section we introduce some combinatorial facts about words that are necessary
    in the proof of the $\frac{4}{5} n$ upper bound on the number of cubes
    in a word of length $n$.
    \begin{lemma}
      \label{lem:1.5p}
      Let $v^3$ and $w^3$ be two nonempty cubes occurring in a word $u$
      at positions $i$ and $j$ respectively, such that:
      $$i\  <\  j \ \le\  i+\left\lceil\frac{|\pr(v)|}2\right\rceil\ .$$
      Then:
      $$|\pr(w)|=|\pr(v)|\  \ \textrm{or}\ \ |\pr(w)| \ge 2\cdot|\pr(v)|-(j-i-1)\ .$$
    \end{lemma}
    \begin{proof}
      Let us denote $p=|\pr(v)|$, $q=|\pr(w)|$, and let $k$ be the position of
      the last letter of $w^3$.
\vskip 0.2cm
      \noindent{\bf Case 1.}\

      \noindent
      Let us first consider the case, when the (last) occurrence of $w^3$ is totally inside $v^3$.
      Observe that $k$ must then be within the last of the three $v$'s,
      since otherwise $w^3$ would occur in $u$ at position $j+p$ or further
      (see also Fig. \ref{pref-struct}).
      Hence, due to Lemma~\ref{lemma3}, we obtain $q=p$.
\vskip 0.2cm
\noindent{\bf Case 2.}

      \noindent
      In the opposite case, let $x$ be the maximal prefix of $w^3$ that lays inside $v^3$.
      If $p \ne q$ then $p+q$ must be greater than $|x|$.
      Indeed, if $p+q\le |x|$ then both $\pr(v)$ and $\pr(w)$ would be subwords
      of $x$, so if $p\ne q$, then one of them would not be primitive
      due to the periodicity lemma.
      Therefore:
      $$ p+q > |x| > |v^3|-(j-i) \ge 3p-(j-i) \ .$$
      Consequently $q \ge 2p-(j-i)+1$.
      \qed
    \end{proof}
    Let us introduce a useful notion of \emph{$p$-occurrence}.

    \begin{definition}
      A \emph{$p$-occurrence} is the (last) occurrence of a cube
      with primitive root of length $p$.
    \end{definition}
    It turns out that the primitive roots of cubes appearing close to each other
    cannot be arbitrary.
    It is formally expressed by the following lemma.

    \begin{lemma}
      \label{l:ppp}
      Let $a_1,a_2,\ldots,a_{p+1}$ be an increasing sequence of positions in a word $u$,
      such that $a_{j+1} \le a_j+p$ for $j=1,2,\ldots,p$.
      It is not possible for all these positions to contain $p$-occurrences.
    \end{lemma}

    \begin{proof}
      Let us assume, to the contrary, that at each of the positions $a_1,a_2,\ldots,a_{p+1}$
      there is a $p$-occurrence.
      Observe that the inequalities from the hypothesis of the lemma imply that
      the primitive roots of cubes occurring at these positions are all
      cyclic rotations of each other.
      There are only $p$ different rotations of such primitive roots;
      therefore, due to the pigeonhole principle, some two of them must be equal.

      It suffices to show that all these cubes have the same length,
      because then some two of them are equal, and consequently one of them is not the last
      occurrence of the cube.

      Assume to the contrary that some of the considered cubes have different lengths.
      Let $a_j$ and $a_{j+1}$ be two considered positions, such that cubes
      ($v^3$ and $w^3$ respectively) occurring at these positions have different lengths
      ($3kp$ and $3lp$ respectively, for $k\ne l$).
      Let us consider two cases.
\paragraph{\bf Case 1.}
      If $l<k$, then $3kp-3lp\ge 3p$, and $w^3$  occurs in $u$ at position $a_{j+1}+p$ or further
      (see Fig.~\ref{lemma4-1}).
    \begin{figure}[th]
      \label{l_less_than_k}
      \begin{center}
        \includegraphics[width=10cm]{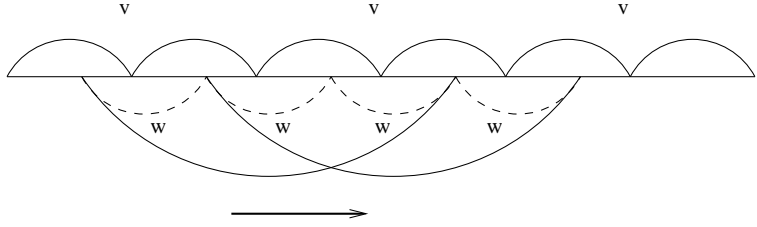}
        \caption{\label{lemma4-1}
          The positions of cubes $v^3$ and $w^3$ in the case $l<k$:
          $a_{j+1}$ is not the last occurrence of $w^3$.
        }
      \end{center}
    \end{figure}

\paragraph{\bf Case 2.}
      If $k<l$, then $3lp-3kp\ge 3p$ and $v^3$ appears in $u$ at position $a_j+p$ or further
      (see Fig.~\ref{lemma4-2}).
    \begin{figure}[th]
      \label{l_greater_than_k}
      \begin{center}
        \includegraphics[width=10cm]{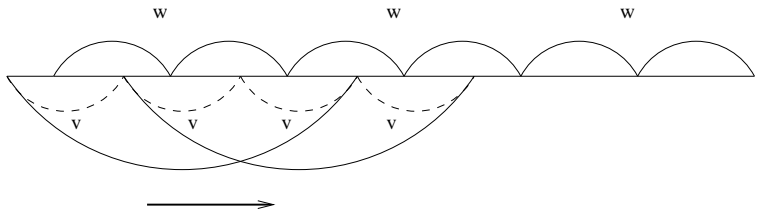}
        \caption{\label{lemma4-2}
          The positions of cubes $v^3$ and $w^3$ in the case $k<l$:
          $a_j$ is not the last occurrence of $v^3$.
        }
      \end{center}
    \end{figure}

      \medskip
      In both cases we obtain a contradiction.
      Hence, it is not possible that the lengths of the cubes differ.
      \qed
    \end{proof}
    Let us introduce a notion of independent prefixes.

    \begin{definition}\label{d:independent_prefix}
      We say that $v$ is the {\em independent prefix} of $u$ if it is
      the shortest prefix of $u$ that is:
      \begin{enumerate}
        \item
          a single letter word, if there is no occurrence of a cube at the first position of $u$,
          or otherwise
        \item
          a prefix that ends with a $q$-occurrence (for some $q\ge 1$)
          followed by exactly $\left\lceil\frac{q}2\right\rceil$
          positions without any occurrences
          (here all occurrences are considered within $u$).
      \end{enumerate}
    \end{definition}
    It is not obvious that the above definition is valid.
    Therefore, we prove the following lemma:

    \begin{lemma}
      \label{lem:ind_prefix_ex}
      For every word $u$, there exists an independent prefix $v$ of $u$.
    \end{lemma}
    \begin{proof}
      If there is no occurrence of a cube at the first position of $u$, then obviously $v=u[1]$.

      In the opposite case, let us assume --- to the contrary --- that the independent prefix does not exist.
      Let $q$ be the maximum such value, that there exists a $q$-occurrence in $u$,
      and let $i$ be the rightmost position in $u$ that contains a $q$-occurrence.
      From Lemma \ref{lem:1.5p}, $\left\lceil\frac{q}2\right\rceil$ positions
      following $i$ do not contain any occurrences of cubes.
      Thus, the prefix $u[1\dotdot i+\left\lceil\frac{q}2\right\rceil]$ satisfies
      the definition of an independent prefix --- a contradiction.
      \qed
    \end{proof}

\section{Algorithm Abstract-Simulation}
      Let $v$ be the independent prefix of a word $u$ and let $|v|>1$.
      Let $(c_i)_{i=1}^{|v|}$ be a sequence describing the occurrences starting within $v$:
      $c_i=0$ iff there are no occurrences in position $u[i]$, and
      $c_i=q$ iff there is a $q$-occurrence in position $u[i]$.
      We start with the following observations.
      \begin{enumerate}[a)]
        \item
          If $c_i$ and $c_j$ is a pair of consecutive nonzero elements of $c$
          (i.e.\ $i<j$, $c_i, c_j > 0$ and $c_{i+1} = \dots = c_{j-1}=0$) then
          $j-i \le \left\lceil \frac{c_i}{2} \right\rceil$.
          Indeed, if $j-i > \left\lceil \frac{c_i}{2} \right\rceil$, then the prefix of $u$ of length
          $i + \left\lceil \frac{c_i}{2} \right\rceil$ or shorter would be an independent
          prefix of $u$.
\vskip 0.1cm
        \item
          For $c_i$ and $c_j$ as in a), $c_j \ge 2c_i-(j-i-1)$.
          This observation is due to Lemma~\ref{lem:1.5p}.
\vskip 0.1cm
        \item
          From Lemma~\ref{l:ppp} and due to a) we have that
          no $q+1$ consecutive positive elements of $c$ are equal to $q$.
      \end{enumerate}

      From now on, we abstract from the actual word $u$, and
      focus only on the properties of sequence $c$.
      We will analyze the ratio $R$ of nonzero elements of $c$ to the length of $c$.

      Let us observe that if $c$ contains such a pair of equal elements $c_i=c_j>0$,
      that all the elements between them are equal zero, then all the elements between
      $c_i$ and $c_j$ can be removed from $c$ without decreasing $R$.
      Also, if $c$ contains a subsequence of consecutive elements
      equal to $q$ ($q>0$) of length less than $q$ then this subsequence can be extended
      to length $q$ without decreasing $R$.
      Let $c'$ be the sequence obtained from $c$ by performing
      the described modification steps (as many times as possible).
      Observe that none of these steps violates properties a)--c).

      Every possible sequence $c'$ can be generated by the (nondeterministic)
      pseudocode shown below.
      The following variables are used in the pseudocode:
      \begin{itemize}
        \item $p$ --- the value of the last positive element of $c'$
        \item $len$ --- the length of the sequence $c'$ without
          $\lceil p/2\rceil$ trailing zeros
        \item $occ$ --- the number of positive elements in $c'$
        \item $l$ --- the gap between consecutive different positive elements
          of $c'$
        \item $\alpha$ --- the difference between the actual value of a positive
          element of $c'$ and the lower bound from Lemma \ref{lem:1.5p}.
      \end{itemize}
      Each step of the \textbf{repeat} loop corresponds to extending sequence $c'$,
      i.e.\ adding $l$ zeros and $p$ elements of value $p$.

      \begin{figure}[h!]
        \centering
        \begin{tabular}{cccccccc}
          3 3 3 & 0 &
          $\underbrace{5\ldots 5}_{5\ \mbox{times}}$ & 0 0 &
          $\underbrace{20\ldots 20}_{20\ \mbox{times}}$ &\ \ $\underbrace{0\ldots 0}_{6\ \mbox{times}}$ &
          \ \ $\underbrace{34\ldots 34}_{34\ \mbox{times}}$ &\ \ $\underbrace{0\ldots 0}_{17\ \mbox{times}}$ \\
        \end{tabular}
        \caption{An example of sequence $c'$.
          The length of the sequence is 88 and it contains 62 positive elements.
          The ratio is $62/88\approx 0.70<4/5$.}
      \end{figure}
      \vskip 0.3cm
      Note that the algorithm specified by the pseudocode is nondeterministic
      in several different aspects --- the initial value of $p$, the number of
      steps of the \textbf{repeat} loop and values of $l$ and $\alpha$.
\myskip
      \begin{center}
        \fbox{
          \begin{minipage}{9cm}
            \flushleft
            \medskip \medskip
            ~~{\bf\large Algorithm Abstract-Simulation}
            \vskip 0.2cm \noindent
            ~~~$p:=$\ some positive integer;\\
            ~~~$occ:=p; \quad len:=p$;\\
            ~~~output: $\underbrace{p\dots p}_{p\ \mathrm{times}}$\vskip 0.2cm \noindent
            ~~~{\bf repeat an arbitrary number of times}\vskip 0.2cm \noindent
            ~~~~~~~~~$\textsl{Invariant}\ I(p,occ,len): \frac{occ}{len+\frac{p}{2}} \ \le\  \frac{4}{5}$.\\
            ~~~~~~~~~$l:=$\ some integer from interval $[0,\left\lceil\frac{p}2\right\rceil)$;\\
            ~~~~~~~~~$\alpha:=$\ some nonnegative integer;\\
            ~~~~~~~~~$p:=2p-l+\alpha$;\\
            ~~~~~~~~~$occ:=occ+p$;\\
            ~~~~~~~~~$len:=len+l+p$;\\
            ~~~~~~~~~output: $\underbrace{0\dots 0}_{l\ \mathrm{times}}\,
                           \underbrace{p\dots p}_{p\ \mathrm{times}}$\\
            \medskip  \medskip
          \end{minipage}
        }
      \end{center}


      \myskip
\section{Upper bound on the number of cubic subwords}
    \begin{lemma}[Invariant lemma]
      \label{l:invariant}
      The following condition $I(p,occ,len)$:
      $$\frac{occ}{len+\frac{p}2} \le \frac{4}{5}$$
      is an invariant of the Abstract-Simulation Algorithm.
    \end{lemma}

    \begin{proof}
      Before the first execution of the \textbf{repeat} loop, $occ=len=p$,
      and consequently $I(p,occ,len)$ holds:
      $$\frac{p}{p+\frac{p}2} \ =\  \frac{1}{\frac{3}2} \ =\ \frac{2}3 \ \le\  \frac{4}{5}\ .$$
      
      Therefore, we only need to prove that if $I(p,occ,len)$
      holds then $I(p',occ',len')$ also holds, where
      $p'$, $occ'$ and $len'$ are the values obtained as a result of
      a single step of the \textbf{repeat} loop, i.e.:
\myskip
\begin{center}
\fbox{
  \medskip
    \begin{minipage}{4cm}
  \medskip
    ~~~$p'=2p-l+\alpha,$

    ~~~$ occ'=occ+2p-l+\alpha,$

    ~~~$  len'=len+2p+\alpha.$
    \medskip
    \end{minipage}
}
\end{center}
      \myskip
      Let us restate $I(p',occ',len')$ equivalently in the following
      way:
      \begin{equation}\label{e:I1}
      5\cdot occ+10p-5l+5\alpha    \ \le \    4\cdot len+8p+4\alpha+4\cdot \frac{2p-l+\alpha}2\ .
      \end{equation}
      On the other hand, $I(p,occ,len)$ can be expressed as
      $$5\cdot occ    \ \le\     4\cdot len+4\cdot \frac{p}2\ .$$
      Hence, in order to show (\ref{e:I1}), it is sufficient to prove that:
      \begin{equation}\label{e:I2}
      10p-5l+5\alpha\     \le  \   8p+4\alpha+2\cdot (2p-l+\alpha)-2p\ .
      \end{equation}
      As a result of some rearrangement, (\ref{e:I2}) can be expressed as
      $$0 \le 3l+\alpha$$
      and this inequality trivially holds.
      \qed
    \end{proof}

    \noindent
    We can now show the upper bound for the number of cubes in independent prefixes.

    \begin{lemma}
      \label{lem:ind_prefix_4_5n}
      Let $v$ be the independent prefix of $u$.
      The number of different nonempty cubes that occur in $u$ and start within $v$
      is not greater than $\frac{4}{5}\cdot|v|$.
    \end{lemma}
    \begin{proof}
      Observe that if $v$ satisfies the first condition of
      Definition~\ref{d:independent_prefix}, then the conclusion trivially holds.
      Therefore, from now on we assume that $|v|>1$.

      As described in the previous section, instead of computing
      the ratio of cubes that occur in $u$ and start within $v$, we can
      deal with the ratio $R$ of nonzero elements of the corresponding sequence
      $c$ to the length of $c$ and show that $R \le \frac{4}5$.
      For this it suffices to prove that for any valid sequence $c'$ the ratio of nonzero
      elements does not exceed $\frac{4}{5}$.

      The Abstract-Simulation Algorithm generates every possible sequence $c'$.
      Hence, in order to prove the $\frac{4}5$ bound, we need to show that inequality
      $$\frac{occ}{len+\left\lceil\frac{p}2\right\rceil}\ \le\  \frac{4}5$$
      holds for every possible execution of the Algorithm.
      But this inequality is a consequence of the fact that $I(p,occ,len)$
      is an invariant of the Algorithm (Lemma \ref{l:invariant}).
      \qed
    \end{proof}
    
    \begin{theorem}
      \label{thm:ind_prefix_4_5n}
      The number of different nonempty cubes that occur in a word of length $n$
      is not greater than $\frac{4}{5}n$.
    \end{theorem}

    \begin{proof}
      We prove the theorem by induction on $n$.
      The basis ($n=0$) is trivial.
      Now assume that the conclusion holds for all words of length not
      exceeding $n$ and consider a word $u$ of length $n+1$.
      Due to Lemma \ref{lem:ind_prefix_ex}, there exists the independent prefix
      $v$ of $u$, $v \ne \varepsilon$, $u=vw$.
      The cubes occurring within $u$ can be divided into two groups: the
      ones that start within $v$ and the ones that occur totally inside $w$.
      By Lemma \ref{lem:ind_prefix_4_5n}, the number of cubes in the first
      group does not exceed $\frac{4}5 |v|$, and by the inductive hypothesis,
      $\cubes(w) \le \frac{4}5\cdot|w|$.
      In total, there are at most
      $$\frac{4}5\cdot|v| + \frac{4}5\cdot|w| \le \frac{4}5\cdot|u|$$
      cubes within $u$ --- this ends the inductive proof.
      \qed
    \end{proof}

\section{Lower bound on the number of cubic subwords}

      A trivial lower bound on the number of different cubic subwords
      is the word $a^n$ with $\left\lfloor \frac{n}3\right\rfloor$ cubic occurrences.
      The table presented in Figure \ref{fig:cubes-example} contains examples
      of some words with higher number of cubic subwords.
      These words have been computed using extensive computer experiments.

    \begin{figure}[h!p]
      \begin{tabular}{r|p{0.1cm}p{8.5cm}|r|r}
        $n$ & & word & $\#$cubes & ratio \\ \hline
        20 & & {\footnotesize \texttt{01110101011011011000}} & 7 & 0.35 \\ \hline
        30 & & {\footnotesize \texttt{000000110110110101101011010101}} & 11 & 0.36 \\ \hline
        40 & & {\footnotesize \texttt{1101101101110111011100010001000100100100}} & 16 & 0.40 \\ \hline
        50 & & {\footnotesize \texttt{11111111110010010010100101001010100101010010101000}}
        & 20 & 0.40 \\ \hline
        60 & & {\footnotesize \texttt{10100101001010010101001010010101001010010101001010
            1001010100}} & 25 & 0.41 \\ \hline

        70 & & {\footnotesize \texttt{00000011011011010110101101010110101101010110101101
            01011010101101010111}} & 30 & 0.42 \\ \hline

        80 & & {\footnotesize \texttt{11011011010110110101101101011010110101011010110101
            011010110101011010101101010111}} & 34 & 0.42 \\ \hline

        90 & & {\footnotesize \texttt{11101101101110110110111011011011101101110110110111
            0110111011011011101101110110111011101110}} & 40 & 0.44 \\ \hline

        100 & & {\footnotesize \texttt{10001010100101010010101001010010101001010010101001
            01001010010101001010010100101010010100101001010111}}
        & 44 & 0.44\\
      \end{tabular}
      \caption{Examples of words with high number of distinct cubic subwords.}
      \label{fig:cubes-example}
    \end{figure}

    Let us proceed to the construction of the $\frac{1}2n$ lower bound.
    For $i \ge 1$, let $p_i$ be the word $0^i10^{i+1}1$.
    Let $q_n$ be the concatenation $p_1p_2 \ldots p_n$.
    Thus, for instance, $q_4=\mathtt{01001001000100010000100001000001}$.

    \begin{lemma}\label{l:lenqn}
      The length of $q_n$ is $n^2+4n$.
    \end{lemma}

    \begin{proof}
      Clearly $p_i$ contains $2i+3$ bits, so
      $$|q_n| \ =\  \sum_{i=1}^n 2i+3 \ =\  n^2+4n\ .$$
      \qed
    \end{proof}

    \begin{lemma}\label{l:cubesqn}
      The word $q_n$ contains exactly
      $$\frac{n^2}2 + \frac{n}2 - 1 + \left\lfloor\frac{n+1}3\right\rfloor$$
      distinct cubes.
    \end{lemma}

    \begin{proof}
      Note that the concatenation $p_ip_{i+1}=0^i10^{i+1}10^{i+1}10^{i+2}1$
      contains the following $i+1$ cubes:
      $$\bigl(0^i10\bigr)^3,\ \bigl(0^{i-1}10^2\bigr)^3,\ \ldots,\ \bigl(010^i\bigr)^3,\ \bigl(10^{i+1}\bigr)^3\ .$$
      Apart from that, in $q_n$ there are $\left\lfloor\frac{n+1}3\right\rfloor$ cubes
      of the form $0^3,0^6,0^9,\ldots$
      Thus far we obtained
      $$\sum_{i=1}^{n-1} (i+1) + \left\lfloor\frac{n+1}3\right\rfloor \ =\ 
        \frac{n^2}2 + \frac{n}2 - 1 + \left\lfloor\frac{n+1}3\right\rfloor$$
      cubes.

      \begin{figure}[h]
      \centering
      \includegraphics[width=.4\textwidth]{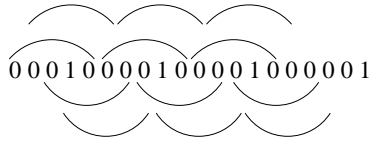}
      \caption{For $i=3$ the word $p_i p_{i+1}$ contains 4 cubes of length $3i+6=15$.}
      \end{figure}

      It remains to show that there are no more cubes in $q_n$.
      Notice that we have considered all cubes $u^3$ for which the number
      of 1's in $u$ equals 0 or 1.
      On the other hand, if this number exceeds 1 then $u$ would contain
      the factor $10^i1$ for some $i \ge 1$ and this is impossible, since
      for a given $i$ such a factor appears within $q_n$ at most twice.
      \qed
    \end{proof}

    \begin{theorem}
      For infinitely many positive integers $m$ there exists a word of length $m$
      for which the number of cubes is $\frac{m}2-o(m)$.
    \end{theorem}

    \begin{proof}
      Due to Lemmas \ref{l:lenqn} and \ref{l:cubesqn}, for any word $q_n$ we have:
      \begin{eqnarray*}
        \frac{|q_n|}2 - \cubes(q_n) \ =\ 
        \frac{n^2}2 + 2n - \frac{n^2}2 - \frac{n}2 + 1 - \left\lfloor\frac{n+1}3\right\rfloor & = & \\
        \frac{3}2 n - \left\lfloor\frac{n+1}3\right\rfloor + 1 \ =\ O(n) & = & o(|q_n|)\ .
      \end{eqnarray*}
      Thus, $\cubes(q_n) = \frac{|q_n|}2 - o(|q_n|)$.
      \qed
    \end{proof}

    Interestingly, the example from the paper \cite{fraenkel-simpson} of a family
    of words that contain $m-o(m)$ squares is quite similar to our example, but
    instead of $p_i$ it utilizes words of the form $p'_i = 0^{i+1}10^i10^{i+1}1$.

  \section{Conclusions}
  In this paper we prove a tight bound for the number of nonprimitive squares
  in a word of length $n$.
  Unfortunately, this does not improve the overall bound of the number of
  squares --- the main open problem is improving the bound for primitive
  squares.

  We also give some estimations of the number of cubes in a string of length
  $n$.
  These bounds are much better than the best known estimations for squares in general.
  We believe that at least the upper bound established in our paper is not tight.

\bibliographystyle{plain}
  \bibliography{cubes}

\end{document}